\documentclass[letterpaper, 10 pt, conference]{ieeeconf}
\usepackage{amsmath,amssymb,graphicx,times}
\IEEEoverridecommandlockouts         

\newcommand{\cmnt}[1]{}

\newtheorem{prop}{Proposition}

\newcommand{\nlnn}{\protect\nonumber\\}

\cmnt{

}

\newcommand{\matx}[1]{\protect\begin{bmatrix} #1 \end{bmatrix}}
\newcommand{\lm}{\lambda}
\newcommand{\bA}{\boldsymbol{A}}
\newcommand{\bB}{\boldsymbol{B}}
\newcommand{\bC}{\boldsymbol{C}}
\newcommand{\bD}{\boldsymbol{D}}

\newcommand{\bM}{\boldsymbol{M}}

\newcommand{\bz}{\boldsymbol{z}}
\newcommand{\AW}{\bA_{\textnormal{W}}}

\newcommand{\ALRW}{\bA_{\textnormal{LRW}}}

\newcommand{\AJPC}{\bA_{\textnormal{JPC}}}
\newcommand{\ALSJPC}{\bA_{\textnormal{LSJPC}}}
\newcommand{\bX}{\boldsymbol{X}}
\newcommand{\bY}{\boldsymbol{Y}}
\newcommand{\bZ}{\boldsymbol{Z}}
\newcommand{\bV}{\boldsymbol{V}}

\newcommand{\bU}{\boldsymbol{U}}

\newcommand{\bR}{\boldsymbol{R}}
\newcommand{\bS}{\boldsymbol{S}}

\newcommand{\bI}{\boldsymbol{I}}
\newcommand{\Id}[1]{\bI_{\!#1}}
\newcommand{\bO}{\boldsymbol{O}}
\newcommand{\Oh}{\mathcal{O}}

\newcommand{\bK}{\boldsymbol{K}}

\newcommand{\ZKL}{\bK_{\!\bZ}}
\newcommand{\ZKLL}{\bK_{\!\bZ,L}}
\newcommand{\bphi}{\boldsymbol{\phi}}

\newcommand{\mymathop}{\operatorname*}
\newcommand{\Ex}{\mymathop{E}}
\newcommand{\rank}{\mymathop{rank}}
\newcommand{\tr}{\mymathop{tr}}
\newcommand{\minm}[1]{\mymathop{minimize}_{#1}\ }
\newcommand{\sbjt}{\mymathop{subject\ to}\ }

\newcommand{\SmX}{\bC_{\!\bX}}
\newcommand{\SmY}{\bC_{\!\bY\!}}
\newcommand{\SmYh}{\bC_{\!\bY}^{1/2}}
\newcommand{\SmYhi}{\bC_{\!\bY}^{-1/2}}
\newcommand{\SmZ}{\bC_{\!\bZ}}

\newcommand{\SmXY}{\bC_{\!\bX\!\bY}}

\newcommand{\VL}{\bV_{\!L}}

\newcommand{\UL}{\bU_{\!L}}

\newcommand{\VXL}{\bV_{\!\bX\!,L}}
\newcommand{\VXLb}{\bV_{\!\bX\!,\overline{L}}}
\newcommand{\VYL}{\bV_{\!\bY\!,L}}
\newcommand{\VYLb}{\bV_{\!\bY\!,\overline{L}}}
\newcommand{\VZL}{\bV_{\!\bZ,L}}
\newcommand{\VZLb}{\bV_{\!\bZ,\overline{L}}}

\newcommand{\SL}{\bS_{\!L}}

\newcommand{\SZ}{\bS_{\!\bZ}}

\newcommand{\VX}{\bV_{\!\bX}}
\newcommand{\VY}{\bV_{\!\bY}}
\newcommand{\VZ}{\bV_{\!\bZ}}
\newcommand{\SZL}{\bS_{\!\bZ,L}}
\newcommand{\SZLb}{\bS_{\!\bZ,\overline{L}}}
\newcommand{\real}{\mathbb{R}}
\newcommand{\compl}{\mathbb{C}}
\newcommand{\field}{\compl}

\newcommand{\minMN}{{\min(M,N)}}

\newcommand{\RYL}{\bR_{\bY\!,L}}

\newcommand{\OmA}{\Omega_{\bA}}

\newcommand{\YB}{\bY_{\!\!\bB}}

\title{\LARGE \bf Well-Conditioned Linear Minimum 
Mean Square Error Estimation}

\author{Edwin K. P. Chong
\thanks{Edwin K. P. Chong is with Dept.\ of Electrical and Computer Engineering, and Dept.\ of Mathematics, Colorado State University, CO 80523. +1-970-491-7858. {\tt\small edwin.chong@colostate.edu}}
}%

\begin{document}
\sloppy 
	
\maketitle
\thispagestyle{empty}
\pagestyle{empty}

\begin{abstract}
Linear minimum mean square error (LMMSE) estimation is often ill-conditioned, suggesting that unconstrained minimization of the mean square error is an inadequate approach to filter design. To address this, we first develop a unifying framework for studying constrained LMMSE estimation problems. Using this framework, we explore an important structural property of constrained LMMSE filters involving a certain prefilter. Optimality is invariant under invertible linear transformations of the prefilter. This parameterizes all optimal filters by equivalence classes of prefilters. We then clarify that merely constraining the rank of the filter does not suitably address the problem of ill-conditioning. Instead, we adopt a constraint that explicitly requires solutions to be well-conditioned in a certain specific sense. We introduce two well-conditioned filters and show that they converge to the unconstrained LMMSE filter as their truncation-power loss goes to zero, at the same rate as the low-rank Wiener filter. We also show extensions to the case of weighted trace and determinant of the error covariance as objective functions. Finally, our quantitative results with historical VIX data demonstrate that our two well-conditioned filters have stable performance while the standard LMMSE filter deteriorates with increasing condition number.
\end{abstract}

\section{Introduction}

We investigate the problem of designing linear filters to minimize the mean square error under a constraint that the filter be well-conditioned. Without this constraint, the optimal filter depends on the \emph{condition number} of the input covariance matrix, often large. This makes the filter numerically unreliable to compute. Incorporating a suitable constraint on the minimization ensures a well-conditioned solution, i.e., the computed filter is reliable. (We use \emph{estimator} and \emph{filter} interchangeably.) 

We first show that, under an appropriate assumption on the constraint set, constrained optimal filters have a particular structure involving a certain \emph{prefilter}, and optimality of a filter is invariant under invertible linear transformations of the prefilter. So, all constrained optimal filters are parameterized by equivalence classes of prefilters. We then consider a well-studied constraint on the \emph{rank} of the filter. Based on the prefilter parameterization, a closed-form optimal solution is easy to derive, called the \emph{low-rank Wiener (LRW)} filter \cite{scharf91}. However, it turns out that constraining the rank still does not ameliorate the problem of ill-conditioning. A related rank-constrained estimator, the \emph{cross-spectral Wiener (CSW)} filter \cite{GoR97}, suffers from the same ill conditioning.

To properly address the ill-conditioning issue, we explicitly constrain the solution to be well-conditioned. No closed-form optimal solution is currently known to exist. Instead, we introduce two approximately optimal solutions: JPC and LSJPC. We analyze their asymptotic performance and show that, just like the LRW filter, they converge to the unconstrained LMMSE solution as the truncation-power loss goes to zero. To illustrate our results, we also show empirical comparisons based on real data: historical daily VIX values.

The practical consequences of ill-conditioning are well-known, but solutions for it in filtering applications are limited.
In \cite{ShR89} and \cite{KuK19}, 
the authors study the numerical behavior related to ill-conditioning of Kalman filters and their extensions. 
In \cite{GhD21}, 
the authors investigate the stability and sensitivity of the condition number of convolution neural network filters for image processing.
Though these studies are related to ill-conditioning, they do not include the design of well-conditioned LMMSE filters of the kind we seek. Because LMMSE filtering is a component of Kalman filters, our filter designs are relevant to Kalman filtering.

Our contributions are summarized as follows:

1. We develop a unifying framework for studying constrained LMMSE estimation problems (Section~\ref{sec:lmmse}) and show that
they all involve a prefiltering structure that is invariant under invertible linear transformations of the prefilter. This parameterizes all such filters by their equivalence classes of prefilters. 

2. We clarify that the rank-constrained optimal solution, the LRW filter (Section~\ref{subsec:LRW}), is generally ill-conditioned (Section~\ref{subsec:wcf}). The same holds for the CSW filter.

3. We introduce two new filters (Section~\ref{sec:filters}), JPC and LSJPC, and show that as their truncation-power loss goes to zero, they converge to the unconstrained LMMSE filter at the same rate as the LRW filter (Section~\ref{sec:diffs}). 

4. We show how to extend our formulation to the case of weighted trace and determinant of the error covariance as objective functions (Section~\ref{sec:wtdet}).

5. We use historical VIX data to demonstrate that the performance of JPC and LSJPC remains stable as we increase the size of the input covariance matrix, while the unconstrained LMMSE filter and LRW filter deteriorate significantly as expected (Section~\ref{sec:performance}). 

\section{LMMSE Estimation}
\label{sec:lmmse}

\subsection{The unconstrained case}

Let $\bX$ and $\bY$ be zero-mean random vectors taking values in $\field^N$ and $\field^M$ respectively.
We wish to estimate $\bX$ from $\bY$ using a linear estimator (or filter) $\bA$ ($N\times M$ matrix) such that the estimate $\bA\bY$ minimizes the mean square error $\Ex[\|\bA\bY - \bX\|^2]$, where $\Ex$ represents either expectation or empirical mean (from data), and $\|\cdot\|$ is the standard $2$-norm. 
In other words, $\bA$ solves the following optimization problem,
called the LMMSE estimation problem:
\begin{equation}
	\minm{\bA} \Ex[\|\bA\bY - \bX\|^2].
	\label{eq:unconstrLMMSE}
\end{equation}

\cmnt{
\begin{figure}
	\begin{center}
		\includegraphics[width=1.5in]{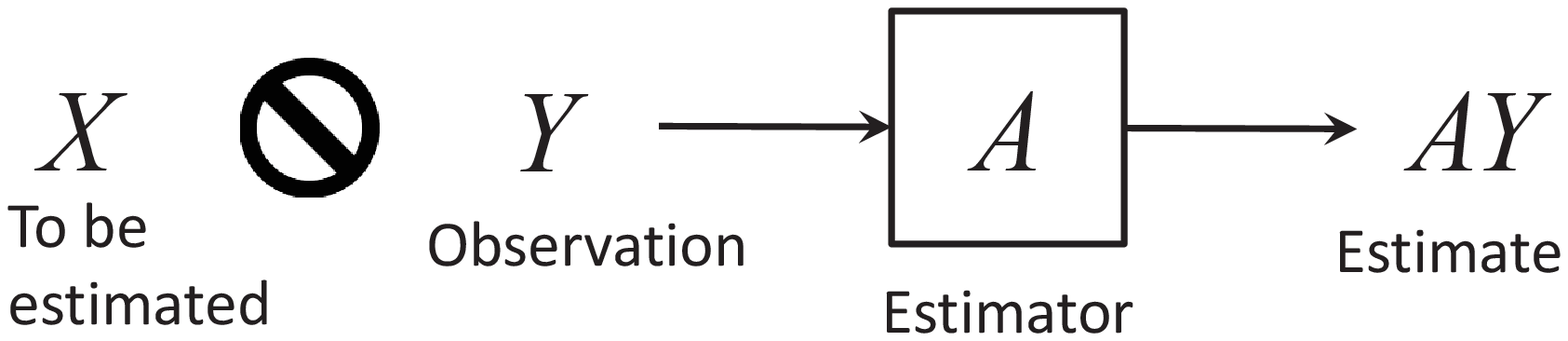}
	\end{center}
	\caption{Linear estimation}
	\label{fig:est}
\end{figure}
}

To express the optimal solution, define the covariance matrix of $\bX$, $\SmX:=\Ex[\bX\bX']$ ($N\times N$), covariance of $\bY$, $\SmY:=\Ex[\bY\bY']$ ($M\times M$), and crosscovariance of $\bX$ and $\bY$, $\SmXY:=\Ex[\bX\bY']$ ($N\times M$), where the superscript prime represents Hermitian transpose. Assume that $\SmY$
has full rank. 
The problem yields a simple closed-form unique solution 
	$\AW :=\SmXY\SmY^{-1}$,
henceforth called the \emph{unconstrained LMMSE filter} (or \emph{Wiener filter}, though the same name appears in other filters too).

\subsection{The constrained case}

Suppose that we now impose an explicit constraint on the estimator $\bA$ in terms of a constraint set $\OmA\subset\field^{N\times M}$. The \emph{constrained LMMSE} problem is then
\begin{align}
	\minm{\bA} & \Ex[\|\bA\bY - \bX\|^2] \label{eq:constrLMMSE}\\
	\sbjt & \bA \in \OmA. \nonumber
\end{align}
Our primary motivation is the difficulty in numerically computing unconstrained LMMSE estimates. Specifically, 
it involves the inverse $\SmY^{-1}$ (i.e., solving linear equations), which often precludes practically computing the solution, for two reasons. If $M$ is very large, then the computational burden might be practically infeasible. 
But more important, the large \emph{condition number} of $\SmY$---the ratio of its largest eigenvalue to its smallest eigenvalue---causes numerical problems.
In this case, we say that the computation is \emph{ill-conditioned}; otherwise, it is \emph{well-conditioned}. 
Ill-conditioning exists regardless of computational execution time or finite-precision arithmetic. 
The problem lies instead in the amplification of errors inherent in solving linear equations.

According to \cite{BeK80}, if the condition number is $10^c$ and the data have $d$ significant figures, then a small data perturbation can affect the solution in the $(d-c)$th place.
In practice, empirical covariance matrices often have significant errors in even the fourth significant figure. 
Moreover, condition numbers easily exceed $10^3$ (e.g., \cite{GhC20}; 
see also Section~\ref{sec:performance}).
In general, the mean condition number of a random $M\times M$ matrix grows with $M$ \cite{Sma85}, \cite{Ede88}. 

\subsection{Structure of optimal constrained estimators}

We need some further definitions. Let $L\leq M$ and $\bB\in\field^{L\times M}$ be full-rank. Suppose that a filter $\bA\in\field^{N\times M}$ factorizes as $\bA=\bD\bB$ for some $\bD\in\field^{N\times L}$. We call $\bB$ a \emph{prefilter} of $\bA$ because the input first goes through $\bB$, and then the prefiltered input goes through $\bD$. Clearly, premultiplying $\bB$ by an invertible matrix also produces a prefilter of $\bA$. The filter $\bA$ is said to be \emph{Wiener-structured} if it has the form $\bA = \SmXY\bB'(\bB\SmY\bB')^{-1}\bB$, i.e., the unique unconstrained LMMSE filter with the prefiltered input. 
Such filters are invariant to premultiplication of $\bB$ by any invertible matrix (this is easy to verify).
Finally, the set $\OmA$ is said to be \emph{Wiener-closed} (with respect to $\bB$) if $\SmXY\bB'(\bB\SmY\bB')^{-1}\bB\in\OmA$. Clearly, $\OmA$ remains Wiener-closed if we premultiply $\bB$ by an invertible matrix. 

For any full-rank $\bB\in\field^{L\times M}$ ($L\leq M$), the special case of $\OmA=\field^{N\times M}$ (unconstrained as in \eqref{eq:unconstrLMMSE}) is Wiener-closed. So are some other constraints of interest. We discuss a well-known example below in Section~\ref{subsec:LRW} 
and introduce another in Section~\ref{subsec:wcf} 
(with an additional requirement). 
Also, the optimal solution to \eqref{eq:unconstrLMMSE} is always Wiener-structured, provided $\bB$ is a prefilter.
In fact, given a Wiener-closed $\OmA$, a filter is optimal if and only if it is Wiener-structured, as shown below.

\begin{prop}\label{prop:optiffWiener}
Let $\OmA$ be Wiener-closed with respect to $\bB$ (full-rank) and let $\bA\in\OmA$ have prefilter $\bB$. Then $\bA$ is optimal for \eqref{eq:constrLMMSE} 
if and only if it is Wiener-structured.
\end{prop}
\begin{proof}
Consider the prefiltered input $\YB:=\bB\bY$. The covariance of $\YB$ is $\bB\SmY\bB'$, which is positive definite because both $\SmY$ and $\bB$ have full rank. The crosscovariance of $\bX$ and $\YB$ is $\SmXY\bB'$. Because $\OmA$ is Wiener-closed, it contains $\SmXY\bB'(\bB\SmY\bB')^{-1}\bB$, which is the unique unconstrained optimal filter with input $\YB$, i.e., has smaller mean square error than all other filters with prefilter $\bB$. Hence, $\bA$ is optimal if and only if $\bA=\SmXY\bB'(\bB\SmY\bB')^{-1}\bB$, which is Wiener-structured. 
\end{proof}

Proposition~\ref{prop:optiffWiener}, though elementary, exposes the intrinsic Wiener-structure (and hence invariance) of all optimal solutions with a Wiener-closed constraint set, thereby parameterizing them by equivalence classes of prefilters. The result also facilitates the design of candidate filters. To wit, once we select a prefilter $\bB$, we can use the filter $\bA = \SmXY\bB'(\bB\SmY\bB')^{-1}\bB$ knowing that any optimal filter with a Wiener-closed constraint has this structure. This form of filters features prominently in the rest of the paper. Later, in Section~\ref{subsec:JPC}, we also exploit the invariance of Wiener-structured filters to invertible linear transformations of $\bB$.

Assume that $\OmA$ is Wiener-closed with respect to each feasible prefilter $\bB$ (i.e., each $\bB$ that is a prefilter for some $\bA\in\OmA$). 
Based on Proposition~\ref{prop:optiffWiener}, \eqref{eq:constrLMMSE} can be posed in an equivalent form involving only Wiener-structured filters:
	\begin{align}
	\minm{\bB} & \Ex[\|\SmXY\bB'(\bB\SmY\bB')^{-1}\bB\bY - \bX\|^2] 
		\label{eq:constrB}  \\
		\sbjt & \bB \mathrm{\ full\ rank},
	\ \SmXY\bB'(\bB\SmY\bB')^{-1}\bB\in\OmA. 
		\nonumber
	\end{align}
This equivalent form of the problem is useful because the decision-variable matrix here is smaller than in \eqref{eq:constrLMMSE}, and the constraint set here is smaller.

\subsection{Rank-constrained filters}
\label{subsec:LRW}

An important special case of \eqref{eq:constrLMMSE} is the following: Given a positive integer $L\leq M$, let 
\begin{equation}
	\OmA = \{\bA\in\field^{N\times M}:\rank(\bA) \leq L\}.
	\label{eq:OmA}
\end{equation}
(If $L\geq \minMN$, then $\OmA = \field^{N\times M}$.)
Clearly, this $\OmA$ is Wiener-closed for any full-rank $\bB\in\field^{L\times M}$. 
Moreover, we can derive a closed-form expression for an optimal solution, in terms of the singular-value decomposition (SVD), using the following notation. Let $\SmYhi$ be the inverse square-root of $\SmY$. Write the SVD $\SmXY\SmYhi=\bU\bS\bV'$ with singular values $\lm_1,\ldots,\lm_\minMN$, listed in descending order by convention.
Let $\UL$, $\SL$, and $\VL$ be the submatrices of $\bU$, $\bS$, and $\bV$, respectively, consisting of the first $\min(L,N)$ columns.
The following \emph{rank-constrained} filter is optimal for \eqref{eq:constrLMMSE} with $\OmA$ in \eqref{eq:OmA}:
\begin{equation}
	\ALRW := \UL\SL\VL'\SmYhi.
	\label{eq:ALRW}
\end{equation}
This solution is derived in \cite{scharf91}, where it is called the \emph{low-rank Wiener (LRW)} filter. Reduced-rank filters have been studied quite extensively; see, e.g., \cite{GoR97}, \cite{HuN01}, 
\cite{Die07}, and \cite{ScC08}.

The LRW filter is Wiener-structured with prefilter $\bB=\VL'\SmYhi$, acting first on $\bY$ with $\SmYhi$, which \emph{whitens} $\bY$---the output $\SmYhi\bY$ has uncorrelated components with unit variance. The LRW filter then applies $\VL'$, which is equivalent to applying $\bV'$ but keeping only the first $\min(L,N)$ components. Because $\SmYhi\bY$ is white, any $\min(L,N)$ of them have the same total power. However, the power lost in applying $\UL\SL$ by keeping only $\min(L,N)$ singular values, an operation called \emph{truncation}, is minimized by picking the first $\min(L,N)$. This is the basic idea of \emph{principal component analysis}.


A closely related cousin of the LRW filter is the \emph{cross-spectral Wiener (CSW)} filter \cite{GoR97}. To define it, let $\bphi_i$ be the $i$th column of $\SmXY\SmYhi$. The quantity $\|\bphi_i\|^2$ is called the $i$th \emph{cross-spectral power} of $\bX$ and $\bY$ \cite{GoR97}. 
Next, suppose that we order the columns of $\bV$ not by the eigenvalues as before but by the cross-spectral powers instead. If we now redefine $\VL$ to be the first $\min(L,N)$ columns of $\bV$ according to this new ordering, then the expression in \eqref{eq:ALRW} gives the CSW filter. 
Because CSW is rank-constrained, and LRW is the unique optimal rank-constrained filter, generally CSW has worse mean square error in theory than LRW. More important, as shown next, neither LRW nor CSW is well-conditioned.

\subsection{Well-conditioned filters}
\label{subsec:wcf}

The LRW and CSW filters involve $\SmYhi$, the inverse of an $M\times M$ matrix $\SmYh$, which is (potentially) large and ill-conditioned. Indeed, the condition number of $\SmYh$ is $\lm_1^{1/2}/\lm_M^{1/2}=(\lm_1/\lm_M)^{1/2}$, the square root of the condition number of $\SmY$. Therefore, we should expect that computing LRW or CSW is numerically problematic. So it turns out that simply constraining the rank is insufficient to guarantee a well-conditioned solution (though to do so was not the original goal of the LRW and CSW filters).

To circumvent this issue, we define a new constraint as follows. We say that a matrix is \emph{$L$-well-conditioned} ($L$-WC) if it can be computed (from $\SmY$ and $\SmXY$) without any inverse larger than $L\times L$. Define $\OmA$ to be the set of all $L$-WC filters, which is Wiener-closed for any $L$-WC full-rank $\bB$. Indeed, given any such $\bB$, $\SmXY\bB'(\bB\SmY\bB')^{-1}\bB\in\OmA$ because the inverse in the expression $\SmXY\bB'(\bB\SmY\bB')^{-1}\bB$ is only $L\times L$. We call the problem with this new $\OmA$ the \emph{well-conditioned LMMSE} problem.

The LRW and CSW filters are no longer feasible here because they involve $\SmYhi$. Moreover,
a closed-form expression for the optimal $L$-WC solution is currently unknown to even exist. The best we can do here is to design \emph{approximately} optimal solutions and analyze their performance. Two such filters follow next.


\section{Well-Conditioned Filter Designs}
\label{sec:filters}

\subsection{JPC filter}
\label{subsec:JPC}

Let $\bZ$ be the vector with $\bX$ stacked above $\bY$. Let
$\SmZ$ be its $(M+N)\times (M+N)$ covariance,
which can be partitioned naturally as
\begin{equation}
	\SmZ = \matx{\SmX & \SmXY \\ \SmXY' & \SmY}.
	\label{eq:SmZ}
\end{equation}
Define the eigendecomposition
$\SmZ = \VZ\SZ\VZ'$
with eigenvalues ordered from largest to smallest as usual.
The matrices $\VZ$ and $\SZ$ (both $(M+N)\times(M+N))$) depend jointly on $\SmX$, $\SmY$, and $\SmXY$.
Next, partition $\VZ$ into $\VX$ (top $N$ rows) and $\VY$ (bottom $M$ rows).
Given $L\leq M$, write
\begin{align}
	\VX &= \matx{\VXL & \VXLb},\ \VY = \matx{\VYL & \VYLb}, \nlnn
	\VZ &= \matx{\VZL & \VZLb},\nlnn
	\SZ &= \matx{\SZL & \bO_{L\times(N+M-L)}  \\ \bO_{(N+M-L)\times L} & \SZLb }
\end{align}
where $\VXL$ is $N\times L$, $\VYL$ is $M\times L$, $\VZL$ is $(M+N)\times L$,
and $\SZL$ is $L\times L$.

Next, define the $L\times M$ \emph{resolution} matrix $\RYL := (\VYL'\VYL)^{-1}\VYL'$,
which produces the coordinates of the orthogonal projection onto the range of $\VYL$. 
Computing $\RYL$ involves the inverse of only $\VYL'\VYL$, which is $L\times L$ and so is well-conditioned by definition. We use $\RYL$ as the prefilter to define the following filter:
\begin{equation}
	\AJPC := \SmXY\RYL'(\RYL\SmY\RYL')^{-1}\RYL, \label{eq:AJPC1}
\end{equation}
which we call the \emph{joint-principal-component (JPC)} filter. JPC is well-conditioned---the matrix inverses are $L\times L$. 

The invertible matrix $(\VYL'\VYL)^{-1}$ premultiplying $\VYL$ in $\RYL$ can be eliminated, simplifying \eqref{eq:AJPC1} to
\begin{equation}
	\AJPC = \SmXY\VYL(\VYL'\SmY\VYL)^{-1}\VYL'.
	\label{eq:AJPC}
\end{equation}
So, though we started with $\RYL$ to explain the approach, we do not need it to implement the filter. We can just substitute the simpler $\VYL'$ for $\RYL$ as the prefilter even though the former is not an orthogonal-resolution matrix. 

The JPC filter is well-conditioned, involving only an $L\times L$ inverse.
Moreover, it is Wiener-structured with prefilter $\VYL'$.
A similar filter was considered recently in \cite{GhC20} but not formally derived there and was evaluated only empirically, suggesting promising performance for JPC, even when $\AW$ fails badly because of ill-conditioning. Our results in Sections~\ref{sec:diffs} and \ref{sec:performance} corroborate this suggestion.

\subsection{Least-squares JPC (LSJPC) filter}

We can simplify the JPC filter by the following observation. First define the Karhunen-Lo\`eve transform of $\bZ$ by $\ZKL := \VZ'\bZ$, so that $\bZ=\VZ\ZKL$. This means that $\bX=\VX\ZKL$ and $\bY=\VY\ZKL$. Next, let $\ZKLL$ be the top $L$-subvector of $\ZKL$, so that $\bX\approx\VXL\ZKLL$ and $\bY\approx\VYL\ZKLL$. We now estimate $\ZKLL$ from $\bY$. However, we cannot use the LMMSE filter 
for this task because it involves $\SmY^{-1}$. Instead, we use a \emph{least-squares} estimate based on $\bY\approx\VYL\ZKLL$, giving the formula $\ZKLL \approx (\VYL'\VYL)^{-1}\VYL'\bY=\RYL\bY$. Again, we recognize this to be the resolution of $\bY$ onto the range of $\VYL$. Using this estimate, we get the simple filter
\begin{equation}
	\ALSJPC := \VXL\RYL = \VXL(\VYL'\VYL)^{-1}\VYL',
\end{equation}
called the \emph{least-squares JPC (LSJPC)} filter. It is well-conditioned and is simpler than $\AJPC$---it involves fewer multiplications. LSJPC is not Wiener-structured. Interestingly, LSJPC involves inverting only $\VYL'\VYL$, suggesting better conditioning of $\ALSJPC$ relative to $\AJPC$. 

\section{Asymptotic Optimality}
\label{sec:diffs}

Define the \emph{truncation-power loss} $\rho_L$ 
as the power lost by truncation (made precise below), which is small by design.
We now show that LRW, JPC, and LSJPC all converge to the unconstrained LMMSE filter as $\rho_L\to 0$, with the same scaling law.
Hence, these filters are ``asymptotically just as good as'' $\AW$, even though
JPC and LSJPC are well-conditioned while LRW is not.

Our analysis uses the \emph{Bachmann-Landau} notation $\Oh(\cdot)$: Given a matrix $\bM(\rho)$ depending on a parameter $\rho\to 0$, $\bM(\rho)=\Oh(\rho)$ means that for some $c$ and all sufficiently small $\rho$, $\|\bM(\rho)\|\leq c\rho$, where $\|\cdot\|$ is 
some submultiplicative matrix norm \cite{ChZ13} (e.g., nuclear norm).
In this case, we say that $\bM(\rho)\to 0$ with a \emph{scaling law} of $\Oh(\rho)$.
Several algebraic rules help to simplify the calculations: If $\bC$ and $\bD$ are bounded (as $\rho\to 0$), then $\bC\Oh(\rho)=\Oh(\rho)$, $(\bC+\Oh(\rho))^{-1}=\bC^{-1}+\Oh(\rho)$, and $(\bC+\Oh(\rho))(\bD+\Oh(\rho))=\bC\bD+\Oh(\rho)$. In our analysis, it suffices to treat only the singular values (or eigenvalues) from $L+1$ onward
as vanishing.

\subsection{LRW filter}

For LRW, $\rho_L = \tr(\bS)-\tr(\SL)$. 
Using $\bU\bS\bV'=\UL\SL\VL'+\Oh(\rho_L)$ 
and the Bachmann-Landau rules, $\AW=\ALRW + \Oh(\rho_L)$, i.e.,
$\ALRW\to\AW$ with scaling law $\Oh(\rho_L)$. Here and below, a simple calculation shows that the mean square error also converges as $\Oh(\rho_L)$.

\subsection{JPC filter}
\label{subsec:JPCasymp}

For JPC and LSJPC, $\rho_L=\tr(\SZ)-\tr(\SZL)=\tr(\SZLb)$. 
Recall that unlike $\VZ$, the columns of $\VYL$ are not orthonormal, i.e., $\VYL'\VYL\neq\Id{L}$ ($L\times L$ identity) in general. So, we make the additional assumption that 
\begin{equation}
	\VYL'\VYL = \Id{L} + \Oh(\rho_L).
	\label{eq:RYLasymp}
\end{equation}
This assumption is reasonable and natural. 
It holds whenever $N$ scales sublinearly with 
$M$ such that $\VXL = \Oh(\rho_L)$, which implies \eqref{eq:RYLasymp}.

Using \eqref{eq:RYLasymp} and the Bachmann-Landau rules, we have
\begin{align}
\SmY^{-1} &= \VY\SZ^{-1}\VY' \nlnn
	&=\VYL\SZL^{-1}\VYL'+\Oh(\rho_L) \\
\SZL &= \VYL'\VYL\SZL\VYL\VYL'+\Oh(\rho_L) \nlnn
	&= \VYL'\SmY\VYL + \Oh(\rho_L). \label{eq:VYLSmYVYL}
\end{align}
Using these and the Bachmann-Landau rules again,
\begin{align}
	\AW
	&= \SmXY\SmY^{-1} \nlnn
	&= \SmXY\VYL\SZL^{-1}\VYL'+\Oh(\rho_L) \nlnn
	&= \SmXY\VYL(\VYL'\SmY\VYL)^{-1}\VYL'+\Oh(\rho_L) \nlnn
	&= \AJPC + \Oh(\rho_L).
\end{align}
So, $\AJPC\to\AW$ with scaling law $\Oh(\rho_L)$, just like $\ALRW$.

\subsection{LSJPC filter}
\label{subsec:LSJPCasymp}

Again using \eqref{eq:RYLasymp} and the Bachmann-Landau rules,
\begin{equation}
	\SmXY = \VX\SZ\VY'=\VXL\SZL\VYL'+\Oh(\rho_L)
\end{equation}
and so
\begin{align}
	\AW
	&= \VXL\SZL\VYL'\VYL\SZL^{-1}\VYL'+\Oh(\rho_L) \nlnn
	&= \VXL\SZL(\Id{L}+\Oh(\rho_L))\SZL^{-1}\VYL'+\Oh(\rho_L) \nlnn
	&= \VXL\VYL'+\Oh(\rho_L) \nlnn
	&= \ALSJPC+\Oh(\rho_L).
\end{align}
So, $\ALSJPC\to\AW$ with scaling law $\Oh(\rho_L)$, just like $\ALRW$ and $\AJPC$.


The analysis above suggests further simplifications to JPC and LSJPC. For example, approximating $\SmY^{-1}$ by $\VYL\SZL^{-1}\VYL$, $\AJPC\approx\SmXY\VYL\SZL^{-1}\VYL'$. Also, approximating $\VYL'\VYL$ by $\Id{L}$, $\ALSJPC\approx\VXL\VYL'$ (no matrix inverse); it could not get any simpler. 

\section{Weighted Trace and Determinant}
\label{sec:wtdet}

\newcommand{\bG}{\boldsymbol{G}}
\newcommand{\MSE}{\operatorname{MSE}}
\newcommand{\JW}{J_{\text{wt}}}
\newcommand{\Jdet}{J_{\det}}
\newcommand{\SmE}{\bC_\text{err}}
\newcommand{\SmXh}{\bC_{\!\bX}^{1/2}}
\newcommand{\SmXhi}{\bC_{\!\bX}^{-1/2}}
\newcommand{\AWLRW}{\bA_{\textnormal{WLRW}}}

Our objective function so far, the mean square error, can also be expressed as the trace of the \emph{error-covariance} matrix $\SmE := \Ex[(\bA\bY-\bX)(\bA\bY-\bX)']$ as $\Ex[\|\bA\bX-\bY\|^2] = \tr(\SmE)$. An immediate generalization of this objective function is the \emph{weighted trace}, $\JW(\bA) := \tr((\bG'\bG)\SmE)$, where $\bG$ is invertible, leading to a problem like \eqref{eq:constrLMMSE} but with objective function $\JW$ \cite{HuN01}.
Clearly, the regular mean square error is a special case of $\JW$
with $\bG=\Id{N}$. In fact, as shown below, the weighted-trace case is \emph{equivalent}---its solution can be obtained from the regular case. 



Rewriting $\JW(\bA) = \Ex[\|(\bG\bA)\bY-(\bG\bX)\|^2])$, the new objective
is simply the previous objective with covariance of $\bY$, $\SmY$, and 
crosscovariance of $\bG\bX$ and $\bY$, $\bG\SmXY$, except that the decision variable $\bA$ is premultiplied by $\bG$. 
Therefore, the optimal solution can be obtained as a special case of the regular (unweighted) constrained LMMSE problem. Indeed, the rank-constrained optimal filter for the weighted-trace case is easy to write down based on \eqref{eq:ALRW}. Similarly,
we can apply JPC and LSJPC to the weighted-trace case with the same asymptotic analyses.


Another objective function of interest is the \emph{determinant} of the error covariance:
$\Jdet(\bA) := \det(\SmE)$.
The associated rank-constrained optimal filter is also studied in \cite{HuN01}, where it is shown that minimizing $\Jdet$ is equivalent to minimizing $\JW$ with $\bG=\SmXhi$. So the minimizer of $\Jdet$ can be obtained as a special case of the minimizer of $\JW$, and hence also of the unweighted mean square error. Accordingly, JPC and LSJPC can also be applied to the determinant case as a special case of the weighted-trace modification described above. 

\section{Quantitative Performance with Real Data}
\label{sec:performance}

\subsection{Overview of VIX data}

To illustrate the performance of JPC and LSJPC relative to LMMSE, we provide empirical results using real data. We also show that LRW suffers from the same ill-conditioning as LMMSE. We do not consider CSW here as its ill-conditioned behavior is already reflected in LMMSE and LRW. 

Our empirical data consists of historical \emph{Cboe Volatility Index (VIX)} daily closing values \cite{SaM19}. VIX is a quantitative indicator of the equity market's expectation for the strength of future changes of the S\&P~500 index (in the United States). The historical VIX data is freely available and suits our purposes because of its abundance. We obtained our data from \cite{vix}. The VIX sequence has been shown to be empirically wide-sense stationary \cite{SaM19}.

For our estimation problem, we take $\bX$ to be the finite sequence of VIX daily closing values over $N$ consecutive days and $\bY$ to be the sequence of VIX values over the immediate prior $M$ consecutive days. So, our estimation problem is to \emph{predict} $N$ consecutive VIX values from the most recent $M$ prior values, with time measured in days. To estimate covariances, we use samples consisting of vectors of VIX values for $M+N$ consecutive days (corresponding to samples of $\bZ$). These vector-valued samples start at every available day from the earliest date until the date before the $M+N$ days ending with the most recent date available. Owing to stationarity, we treat these samples to be drawn from a common $(M+N)$-variate distribution and with constant mean. The samples are correlated because of the relatively long-range correlation of VIX data.

\subsection{Data processing}

Our VIX dataset corresponds to $7923$ consecutive trading days, starting on 2-January-1990 and ending on 17-June-2021. In our experiments, we vary $M$ from $400$ to $3200$, and we fix $N = 7$. Therefore, the number of data samples is $7923-(M+7)$. We reserve 20\% of the samples for test and evaluation, while the other 80\% are for computing the covariance matrix of $\bZ$ (training); i.e., we conduct \emph{out-of-sample} test experiments. The test-and-evaluation data vectors are sampled uniformly from the available data samples.

For example, for $M=2000$, there are $5917$ vector samples; $4733$ samples are for covariance estimation and $1184$ are for the prediction experiments. Despite the correlation of the samples, these large numbers allow for relatively accurate estimation of the covariance and quantitative performance. We subtracted the empirical average value (approximately $20$) from the VIX values before computing predictions.

Because of the abundance of data, it suffices to compute the empirical covariance matrix of $\bZ$ using the standard estimation formula. To elaborate, let $\bz_1,\bz_2,\ldots,\bz_K\in\real^{M+N}$ be the data-vector samples (with average subtracted). Then, $\SmZ$ is computed using 
	$(\sum_{i=1}^K \bz_i\bz_i')/(K-1)$.
We extract $\SmY$ and $\SmXY$ from $\SmZ$ as submatrices (see \eqref{eq:SmZ}).


\subsection{Performance}

For the remainder of the evaluation, we need to define the \emph{normalized root-mean-square (RMS) error} as a performance metric. The RMS error, a standard performance metric for prediction, is simply the square root of the mean square error. Empirically, we first compute the squared Euclidean norm of the difference between the true and predicted $\bX$ vectors (of length $N=7$). Then we average the squared norm of the errors over all the data samples reserved for test and evaluation (described earlier). None of these samples were used in empirically computing $\SmZ$. We then take the square root of this average to get the RMS error value. The \emph{normalized} RMS error is then calculated by dividing the RMS error value by the RMS value of the $N$-vectors being estimated (including the nonzero mean). This RMS value is empirically computed by taking the squared norm of each sample of $\bX$ plus its empirical mean, averaging these values, and then taking the square root of the average. Using the earlier notation, the normalized RMS error of a filter $\bA$ (LMMSE, LRW, JPC, or LSJPC) is 
\[
	\sqrt{\frac{1}{K}\sum_{i=1}^K \|\bA\bz_i(\bY)-\bz_i(\bX)\|^2}\Bigg/\sqrt{\frac{1}{K}\sum_{i=1}^K \|\bz_i(\bX)+\bar{\bz}\|^2},
\]
where $\bz_i(\bY)$ is the $M$-subvector of $\bz_i$ corresponding to $\bY$, $\bz_i(\bX)$ is the $N$-subvector corresponding to $\bX$, and $\bar{\bz}$ is the average $N$-vector. (Of course, we can dispense with the factor $1/K$.) As intended, the normalization provides performance-metric values that are directly comparable as we vary the parameters in our experiment. Good performance values are significantly smaller than $1$.

Figure~\ref{fig:1}a shows the condition number of $\SmY$ as $M$ increases. As we can see, the condition number increases by over two orders of magnitude as $M$ varies from $M=400$ to $M=3200$. As shown later, above about $M=1600$, computing the inverse of $\SmY$ is unreliable, which corresponds to a condition number of roughly $2\times 10^5$.

\begin{figure}
	\begin{center}
		\includegraphics[width=1.65in]{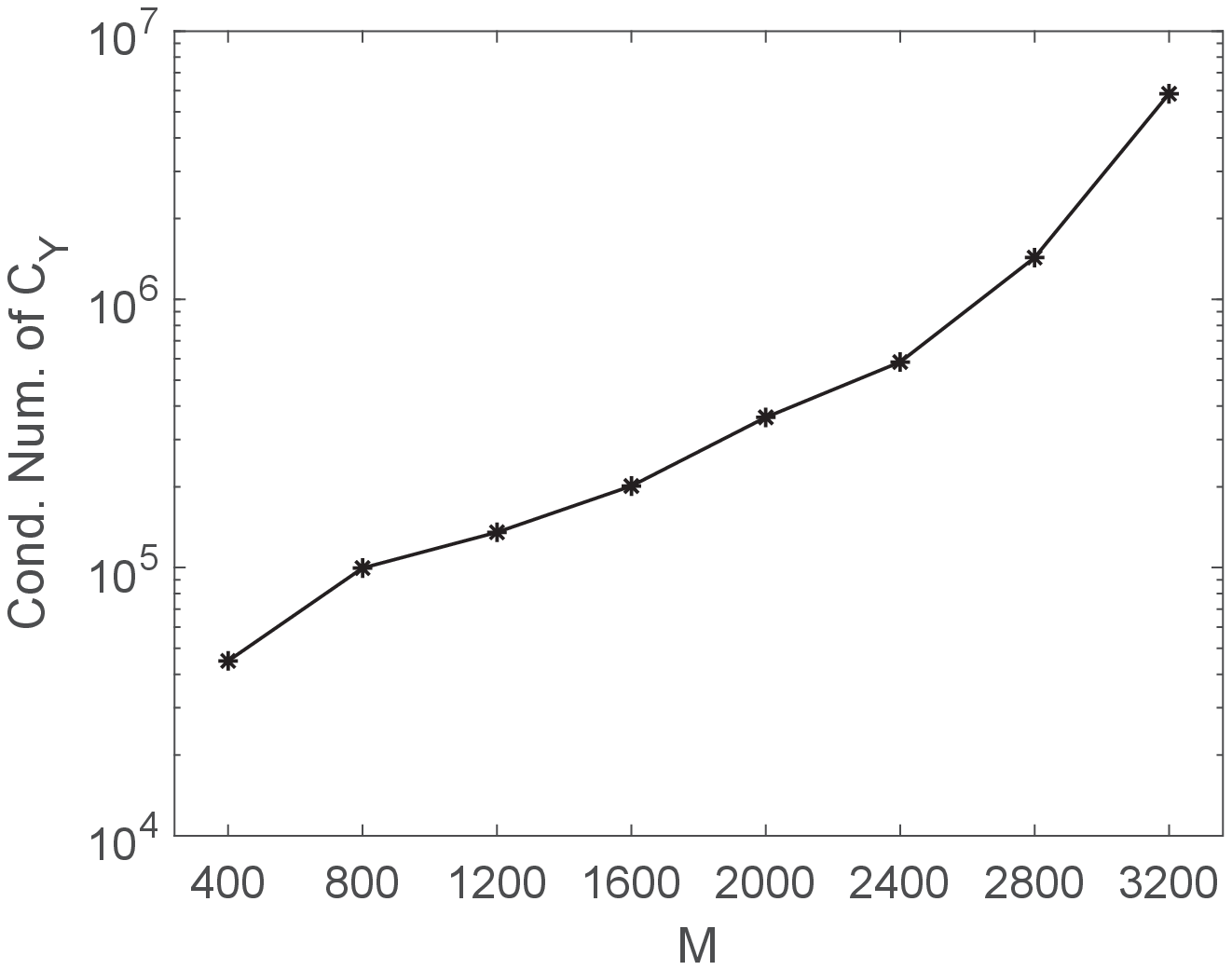}
		\includegraphics[width=1.65in]{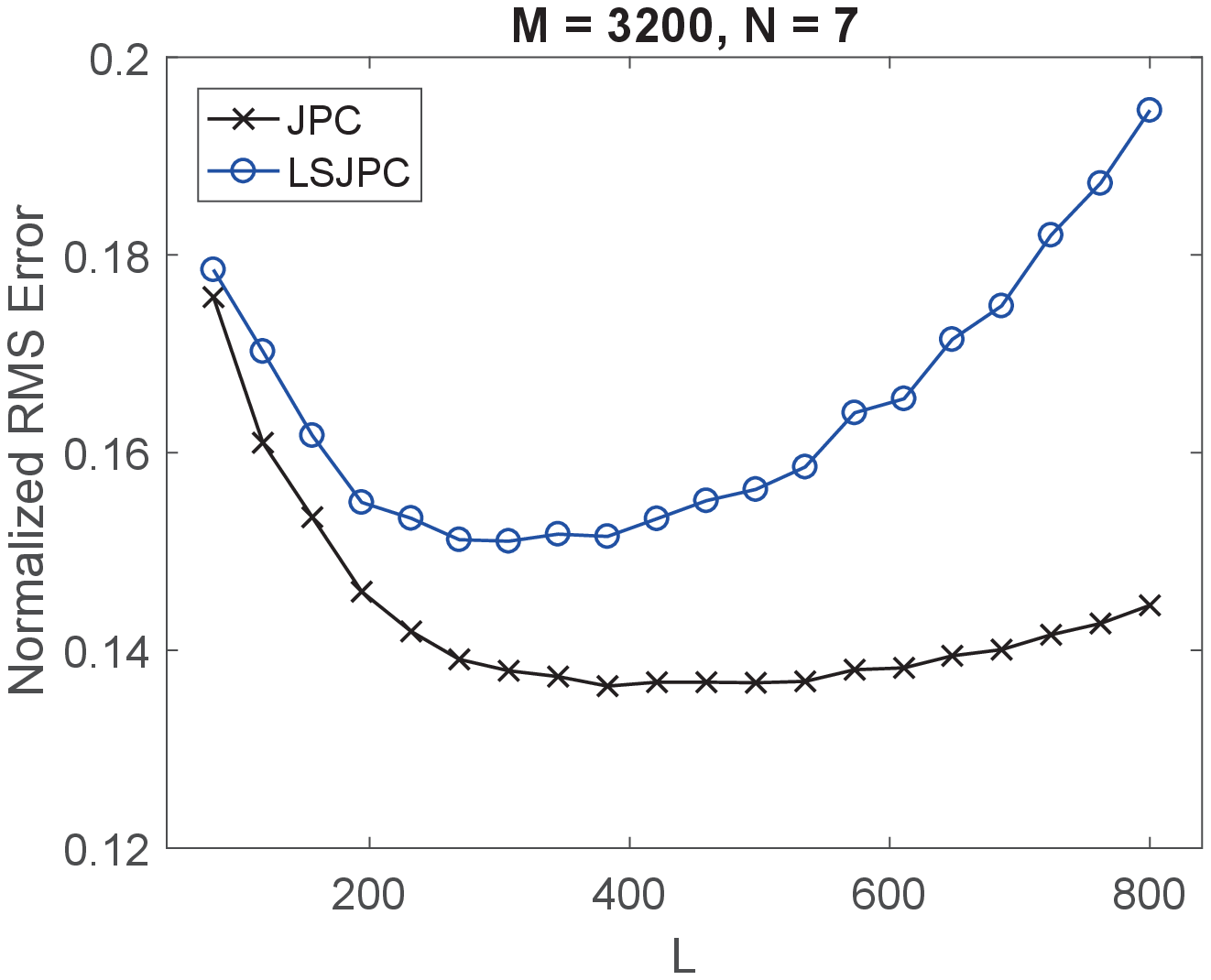}
		\\
		{\small \ \hfill a. \hfill\ \hfill b. \hfill\ }
	\end{center}
		\vspace*{-4mm}
	\caption{a. Condition number of $\SmY$ vs.\ $M$. b. Normalized RMS error vs.\ $L$ for $M=3200$ and $N=7$.}
	\label{fig:1}
		\vspace*{-3mm}
\end{figure}


Figure~\ref{fig:1}b shows the normalized RMS error as a function of $L$ for $M=3200$. As we increase $L$, the truncation-power loss for each filter decreases, and the filter becomes more like $\AW$. At the same time, the ill-conditioning increases, and the filter increasingly exhibits numerical unreliability. 
Thus, there is an optimal intermediate value of $L$ (see Figure~1b).

Practically, a suitable value of $L$ can be found using a simple line-search procedure \cite{ChZ13} 
together with the formula $\tr(\SmX-2\SmXY\bA'+\bA\SmY\bA')$ 
for the mean square error of any filter $\bA$.
Fortunately, with respect to the normalized RMS error, the performance is relatively insensitive close to the minimizer. For JPC in Figure~\ref{fig:1}b, we can vary $L$ by even $200$ without significantly changing the normalized RMS error. This approximately corresponds to a 50\% variation in $L$, which is a very wide margin. LSJPC is slightly more sensitive---we can vary $L$ by about $100$ without significantly affecting the performance, still a very wide margin. This means that the performance is relatively robust to our choice of $L$ (within certain generous bounds), a desirable feature of JPC and LSJPC.
The insensitivity increases as $M$ decreases.
Unsurprisingly, JPC slightly outperforms LSJPC (by less than 10\% at their optimal points), likely because of the additional approximations involved in LSJPC. But recall that LSJPC involves fewer computations 
than JPC, so this tradeoff is favorable in many cases.

\begin{figure}
	\begin{center}
		\includegraphics[width=1.65in]{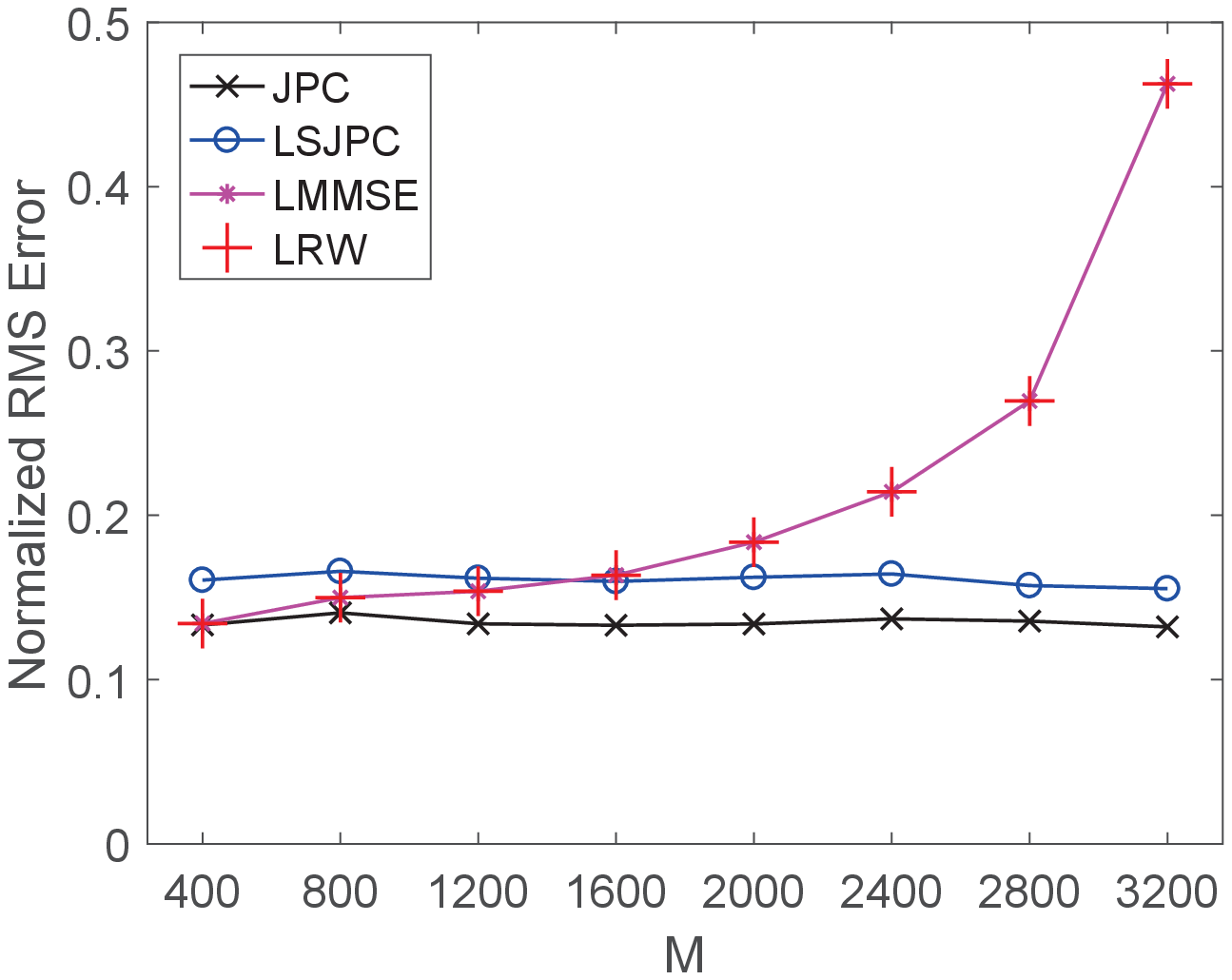}
		\includegraphics[width=1.65in]{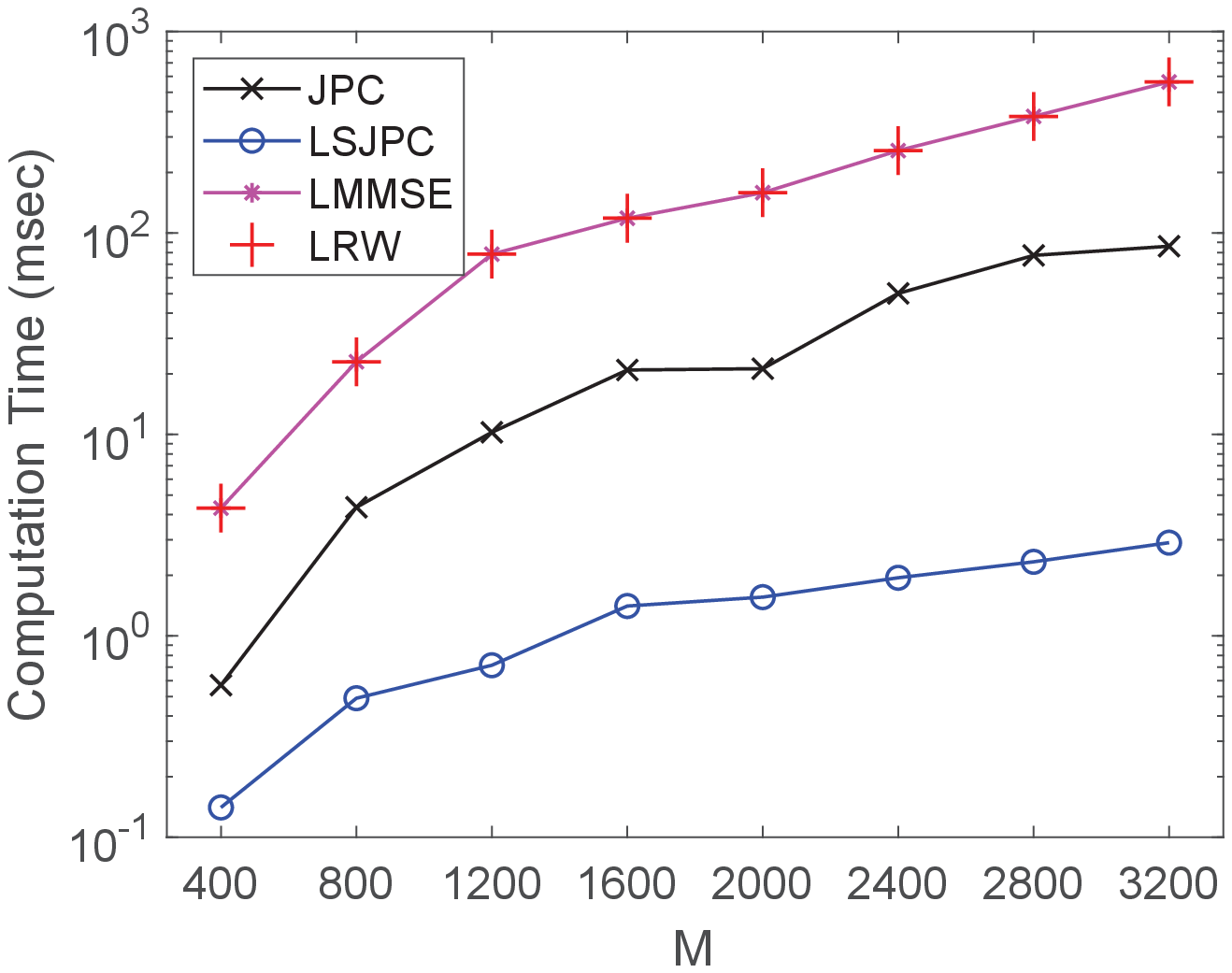}
		\\
		{\small \ \hfill a. \hfill\ \hfill b. \hfill\ }
		\vspace*{-4mm}
	\end{center}
	\caption{
	a. Normalized RMS error vs.\ $M$.  b. Computation time vs.\ $M$.}
	\label{fig:2}
		\vspace*{-3mm}
\end{figure}

\cmnt{
Figure~\ref{fig:argminLvsM} shows plots of the best $L$ values for JPC and LSJPC as we vary $M$. The two dashed straight lines are the lines of best fit through the empirically computed optimal values of $L$, which are naturally somewhat noisy because of the insensitivity mentioned above and also because of noise in the data. Nonetheless, the trend is clearly linear within the insensitivity bounds. The condition numbers associated with these best values of $L$ are relatively small. For example, looking back at Figure~\ref{fig:condVLCYVL2000} (where $M=2000$), these condition numbers are on the order of $10^5$ for JPC and less than $10$ for LSJPC. From Figure~\ref{fig:condCYvsM}, we can see that for LMMSE, the associated condition number (of $\SmY$) is also on the order of $10^5$. However, as shown next, JPC and LSJPC outperform LMMSE at $M=2000$ and higher.

\begin{figure}
	\begin{center}
		\includegraphics[width=1.4in]{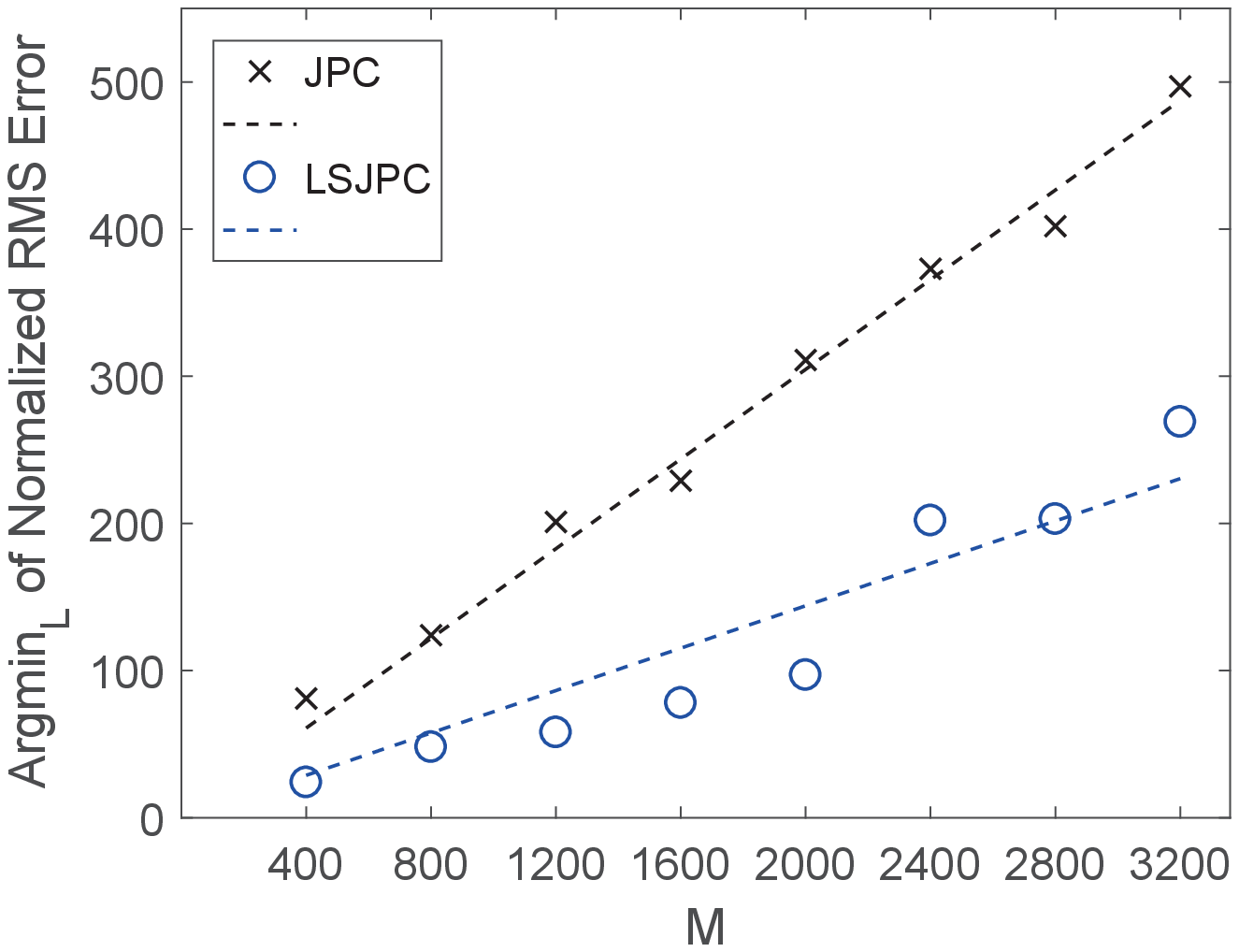}
	\end{center}
	\caption{Best $L$ value as a function of $M$.}
	\label{fig:argminLvsM}
\end{figure}
}

Figure~\ref{fig:2}a 
shows the normalized RMS error as a function of $M$ for JPC, LSJPC, LMMSE, and LRW (indistinguishable from LMMSE). For JPC and LSJPC, we used approximate best $L$ values.
Because $N$ is small in this case and $L\geq N$, $\rank(\ALRW)\leq N$; i.e., the truncation to $L$ in $\ALRW$ has no effect on its performance, which should be close to that of LMMSE except for the impact of computing $\SmYhi$ twice (in LRW) instead of $\SmY^{-1}$ (in LMMSE). We include LMMSE and LRW here to illustrate their performance deterioration as $M$ increases: While JPC and LSJPC have stable performance as $M$ increase, the performance of LMMSE and LRW deteriorates significantly. Moreover, LRW and LMMSE are indistinguishable in performance, reflecting the common cause of the deterioration: the unreliability of computing $\SmY^{-1}$ and $\SmYhi$, respectively.

\cmnt{
\begin{figure}
	\begin{center}
		\includegraphics[width=1.6in]{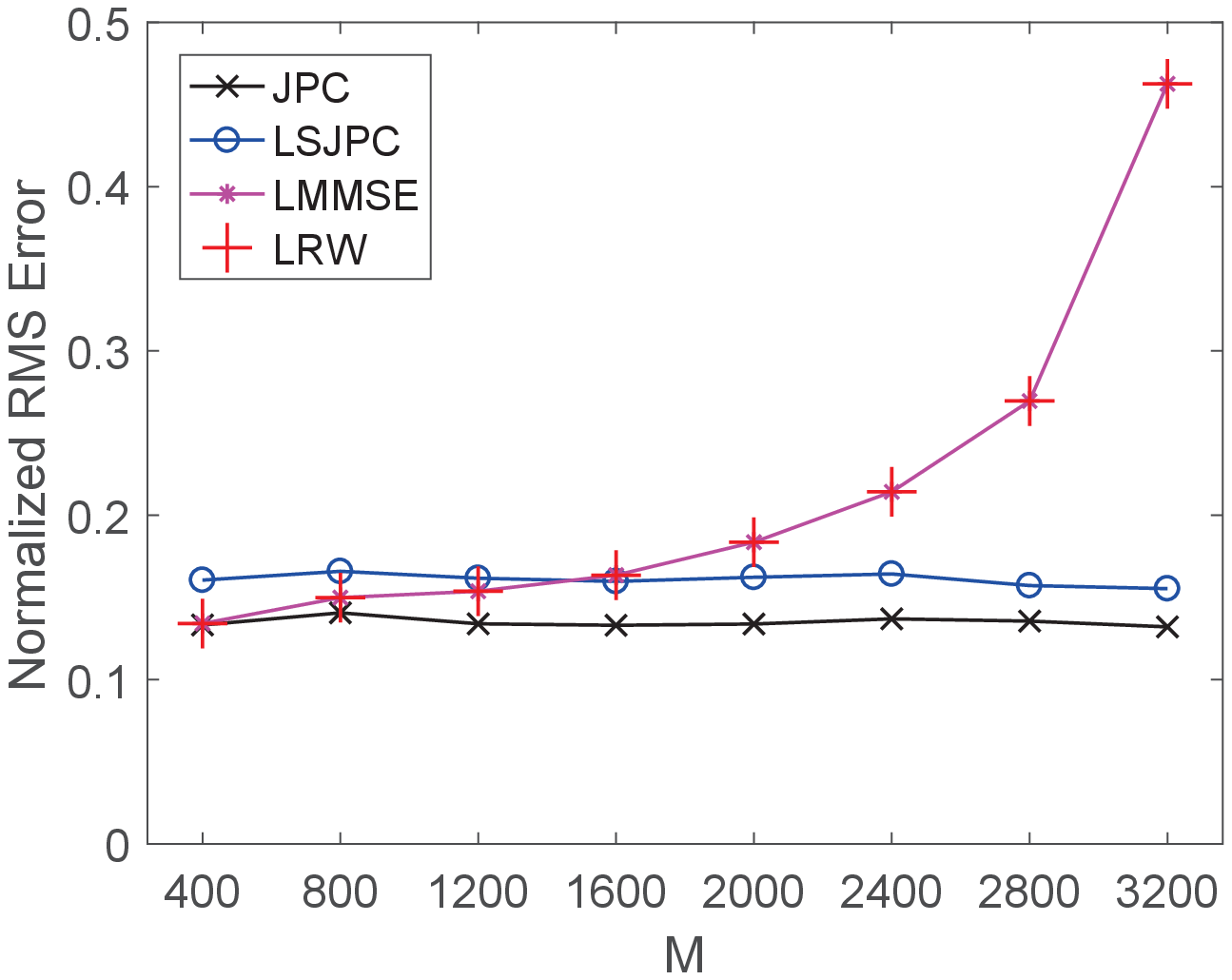}
	\end{center}
	\caption{Normalized RMS error as a function of $M$.}
	\label{fig:RMSerrorvsM}
\end{figure}
}

For small $M$, the performance of JPC and LSJPC are comparable to LMMSE and LRW. However, while the performance of JPC and LSJPC continues to decrease with $M$ (albeit only slightly), the same is untrue for LMMSE and LRW. Above about $M=1600$, LMMSE and LRW are worse than both JPC and LSJPC, indicating the unreliability of computing the inverses $\SmY^{-1}$ and $\SmYhi$. This behavior can be expected of CSW too, which unlike LRW is suboptimal. At $M=3200$, the normalized RMS errors of LMMSE and LRW are approximately \emph{three times} those of JPC and LSJPC, and are increasing steeply. This illustrates the effectiveness of JPC and LSJPC in addressing the ill-conditioning in LMMSE estimation for large $M$.

Note that we could have simply started with a small $M$ and applied LMMSE or LRW, but we would not have known this beforehand. A distinct advantage of JPC and LSJPC is that they have stable performance without knowing how large we can set $M$ for LMMSE or LRW to perform well without ill-conditioning.

Finally, Figure~\ref{fig:2}b shows the computation times for the four filters. Unsurprisingly, JPC and LSJPC take less time than LMMSE and LRW. We can see the tradeoff mentioned earlier between JPC and LSJPC when considering the performance and computation times.

\section{Conclusion}

We have shown that all optimal filters with a Wiener-closed constraint set have an inherent structure: They are parameterized by their equivalence classes of prefilters. The LRW filter, while constrained in its rank, is not well-conditioned, nor is CSW. We introduced two well-conditioned filters, JPC and LSJPC, and showed that both are asymptotically equivalent to the unconstrained LMMSE filter as the truncation-power loss goes to zero. Our results also apply to minimizing weighted trace and determinant of the error covariance. We tested JPC and LSJPC on real data, with promising results.



\end{document}